\documentclass[AMA,Times1COL]{WileyNJDv5} 

\articletype{Article Type}%

\received{Date Month Year}
\revised{Date Month Year}
\accepted{Date Month Year}
\journal{Journal}
\volume{00}
\copyyear{2023}
\startpage{1}

\usepackage{subcaption}

\begin{document}

\title{A Novel Urban Flood Dynamical System Model and a Corresponding Nonstandard Finite Difference Method}

\author[1]{Yongfu Tian}

\author[1]{Shan Ding}

\author[1]{Guofeng Su}

\author[1]{Jianguo Chen}

\authormark{TAYLOR \textsc{et al.}}
\titlemark{A Novel Urban Flood Dynamical System Model and a Corresponding Nonstandard Finite Difference Method}

\address[1]{\orgdiv{School of Safety Science}, \orgname{Tsinghua University}, \orgaddress{\state{Beijing}, \country{China}}}

\corres{Jianguo Chen, Lvdalong Building, Tsinghua University, Beijing, 100084, China. \email{chenjianguo@mail.tsinghua.edu.cn}}


\abstract[Abstract]{Urban flood disaster is one of the most serious natural disasters. Numerous flood simulation models have been proposed and relatively matured. However, two major challenges persist: excessive simplification of the city system and high computational complexity. To break these limitations, this paper develops an Urban Flood Dynamical System Model (UFDSM) based on the concept of the Cellular Automata Urban Flood Model. This model allows flexible customization of cell types and selection of water motion or distribution rules based on actual urban environments to incorporate as much the urban system data as possible. The water motion and distribution rules can be simple, which could reduce the computational complexity, but not arbitrary. So, a sufficient condition is provided so that solutions of dynamical system align with macroscopic physical conditions governing water movement. Then, to preserve the evolutionary properties of the UFDSM, we propose a first-order conservation nonstandard finite difference algorithm. This numerical method ensures positive solutions and conservation of water while maintaining the same fixed-point characteristics as the dynamical system. And, this numerical method is validated by comparing it with an analytical solution.Furthermore, to verify the applicability of our model, we performed an urban flood simulation experiment and compared it to HEC-RAS. There is approximately a 2mm discrepancy in distance $d_p^\prime$ and 0.02mm discrepancy in distance $d_2^\prime$ , with the relative distance $R_p$ about 7.5\% and the relative distance $R_2$ approximately 0.06\%. Additionally, the proposed model is easily coupled with other hydrological processes and facilitates data assimilation, thereby offering promising practical applications.}


\maketitle

\renewcommand\thefootnote{}
\footnotetext{\textbf{Abbreviations:} UFDSM, Urban Flood Dynamical System Model; APC, antigen-presenting cells; IRF, interferon regulatory factor.}

\renewcommand\thefootnote{\fnsymbol{footnote}}
\setcounter{footnote}{1}

\section{Introduction}\label{sec1}

Flood disaster is one of the major natural disasters in the world \cite{cavallo2013}. The frequency and intensity of global flood disasters exhibit an upward trend due to the impact of global climate change \cite{ipcc2021}. Furthermore, in conjunction with urbanization and industrial agglomeration, the severity of urban flood disasters has escalated significantly \cite{adrem2021}. Hence, effective management of urban flood disasters assumes paramount importance. Urban process simulation stands out as a key technology that finds application in urban planning and design urban reconstruction and expansion, flood disaster forecasting, personnel evacuation strategies, and emergency decision-making. Currently, available flood process simulation models such as Mike Urban, TUFLOW, LISFLOOD, HEC-RAS, and INFOWORKS are relatively mature software options \cite{jain2018}. Most of these software employ finite difference or finite volume methods to solve two-dimensional shallow water equations for simulating surface flooding processes with commendable efficacy \cite{guo2021}. Collectively referred to as traditional two-dimensional hydrodynamic models herein.

However, when it comes to simulating complex urban flooding processes using traditional two-dimensional hydrodynamic models encounter two bottlenecks: 

	\begin{enumerate}[1.]
	\item high computational complexity \cite{bulti2020}; 
	
	\item To adapt to the solving process, it is necessary to greatly abandon the urban system data, such as building, road data, and many other urban data that are rich and relatively easy to obtain.
	
	\end{enumerate}

To reduce complexity and enhance computational efficiency, a novel urban flood model called the Cellular Automaton Urban Flood Model (CAUFM) has emerged in recent years \cite{caviedes2018,  dottori2010,   guidolin2016, jamali2019, kassogue2017, yao2021}. This model divides urban areas into discrete cells where water is transported or distributed based on predefined rules. The water motion or distribution rules are designed to emulate real-world flow processes by incorporating physical laws such as Manning's formula, weir flow formula, hole flow formula, and simplified shallow water equations \cite{bates2010, caviedes2018}, as well as intuitive algorithms like weight-based water transport algorithm \cite{guidolin2016} and excess volume spreading rules \cite{jamali2019}. In this model, the urban ground can be represented by diverse cell types \cite{yao2021} with different runoff process. Furthermore, this model readily integrates with other hydro-hydraulic processes such as infiltration, drainage, and river dynamics \cite{wijaya2023}. The model demonstrates remarkable flexibility, aligning with the desired property and developmental trajectory of contemporary urban hydrology and hydrodynamical models \cite{salvadore2015}.

Current research on CAUFM primarily focuses on refining the design of water motion or distribution rules, especially the introduction of momentum \cite{chang2022}. However, some theoretical and numerical problems of CAUFM have not been solved, and its application potential in urban flooding has not been fully explored:

	\begin{enumerate}[1.]
	\item Why are these rules feasible? Which rules are not feasible?
	
	\item  How can the traditional CAUFM updating strategy be improved to address its poor numerical performance, such as oscillation, instability, negative water depth, and non-conservation of water volume in urban flood simulation?
	
	\item The primary advantage of CAUFM is that traditional two-dimensional hydrodynamic models greatly simplify urban data to adapt to algorithms, while CAUFM greatly absorb urban data through simplifying algorithms. This important potential has not yet been explored in current research.
	
	\end{enumerate}

The CAUFM can be considered as a discrete dynamical system with two core elements: rules and time step. The discrete dynamical system can be transformed into a continuous dynamic system if the time step is close to 0 for the CAUFM preferring physical rules. Inspired by this idea, we propose an Urban Flood Dynamical System Model (UFDSM), which calculates the continuous evolution process of water depth in cities based on designed water flow rules, thus constituting a continuous dynamical system. Starting from this continuous dynamical system allows for convenient analysis of the model's evolution. To determine whether the selected rules are reasonable, this paper provides a sufficient condition based on macroscopic evolution laws and the physical limits of urban floods. Additionally, instead of using Cellular Automata's time step updating strategy similar to the explicit Euler method which exhibits poor numerical performance, we design a non-standard finite difference algorithm (NSFD) to solve the UFDSM. This method ensures the preservation of positivity in water depth and conservation of water volume while maintaining consistency between numerical fixed points and fixed points within the dynamical system.

The paper is structured as follows: Section \ref{sec2} constructs the theory of UFDSM and provides constraint conditions for cellular rules to ensure that the model adheres to actual physical constraints of urban flood evolution. Section \ref{sec3} presents a first-order conservative nonstandard finite difference algorithm for solving this type of model, discussing numerical stability conditions and elementary stable conditions of the method. In section \ref{sec4}, we apply the proposed urban flood dynamic system model and numerical method to two constructed examples, comparing simulation results with analytical solutions and HEC-RAS simulations.

\section{Construction of the theory of UFDSM}\label{sec2}

The urban region under consideration is regarded as a bounded closed region $D\in R^2$ with water depth distribution $h\left(x,y\right)\geq0$ and elevation distribution $Z(x,y)\in\ R$, both of which are integrable. The region $D$ is divided into $n$ small regions referred to as cells. The average water depth and its evolution variation over time in the $i-th$ cell are respectively:

\begin{equation}
	\label{eq:2.1}
	\begin{gathered}
		h_i=\frac{1}{S_i} \int_{{cell }_i} h(x, y) d x d y , \\
		\frac{d h_i}{d t}=\frac{F_i(t)}{S_i},
	\end{gathered}
\end{equation}
where $h_i$, $S_i$ and $F_i(t)$ are respectively the average water depth, area and total water inflow flux of the $i-th$ cell. 

In the urban physical field, the total flux $F_i$ in the equation \ref{eq:2.1} is governed by both physical laws and the conditions of the urban system. Traditional two-dimensional hydrodynamic models neglect the wealth of easily accessible urban data, prioritizing algorithms over data. Conversely, CAUFIMs simplify physical equations significantly while incorporating as much urban system data as possible by defining multiple cell types and diverse water interaction rules.
Based on a review of existing CAUFIMs and drawing inspiration from Cellular Automata, we have derived two fundamental hypotheses that serve as the theoretical foundation for constructing the UFDSM presented in this paper.

\textbf{Hypothesis H1:} The water flux between two neighboring cells is solely determined by the water depth and the link properties of the two cells, implying that

\begin{equation}
	F_i\left(t\right)=F\left(h_i,h_{N_i},P_{N_i}\right), \nonumber
\end{equation}
where $N_i=\left[{cell}_{i_1},{cell}_{i_2},\ldots,{cell}_{i_k}\right]$ is the neighboring cells index set of the $i-th$ central cell, $h_{N_i}=\left[h_{i_1},h_{i_2},\ldots,h_{i_k}\right]$, where $i_1, i_2,\ \ldots,i_k\in N_i$, and $P_{N_i}$ is the set of link properties of the central cell and its neighboring cells, where the link properties may include elevation, common-side length, area, etc.

\textbf{Hypothesis H2:} The flux across each boundary of the cell is mutually independent, and the total flux is obtained through the linear superposition of individual boundary fluxes, thereby implying:

\begin{equation}
	F_i\left(t\right)=\sum_{j\in N_i}-f_{ij}\left(h_i,h_j,P_{ij}\right), \nonumber
\end{equation}
where $P_{ij}$ is the link properties between ${Cell}_i$ and ${Cell}_j$, $f_{ij}$ is the water outflow flux from ${Cell}_i$ to ${Cell}_j$, which is a bound function.

The Hypothesis H1, derived from the concept of Cellular Automata, constitutes one of the fundamental distinguishing features of the UFDSM in contrast to other hydrodynamic models. It is assumed that the temporal variation in water depth within the central cell is only determined by the water depth and connection properties of the central cell and its neighboring cells. The predominant driving force for water flow arises from disparities in water levels or gravitational potential energy. Negligible attention is given to momentum due to numerous frictions and obstacles present within complex urban environments, resulting in a conversion of kinetic energy into potential energy to some extent. Building upon Hypothesis H1 and Hypothesis H2, equation \ref{eq:2.1} can be reformulated as a dynamical system characterized by:

\begin{equation}
	\label{eq:2.2}
	\begin{gathered}
		\frac{d h_i}{d t}=\frac{1}{S_i} \sum_{j \in N_i}-f_{i j}\left(h_i, h_j, P_{i j}\right)=\frac{A_i \mathbf{1}}{S_i}, \\
		h_i(0) \geq 0 ,
	\end{gathered}
\end{equation}
where $A=\left(a_{ij}\right)=\left\{\begin{matrix}-f_{ij},&j\in N_i\\0,&j\notin N_i\\\end{matrix}\right.$, which is termed as rules matrix, $A_i=\left[a_{i1},a_{i2},...,a_{in}\right]$ and $ \mathbf{1}=[1,1,\ldots,1]T$. In the subsequent discussion, we assume that the dynamical system \ref{eq:2.2} is well-posed given matrix A, although proving this may challenge. The selection of rules matrix A can be flexible based on the actual urban environment but should have certain limitations to approximate the urban flood process as discussed below. 

From a macro perspective, water movement within a closed urban system should adhere to fundamental physical constraints such as non-negative average water depth, conservation of total water volume, non-increase in total gravitational potential energy of water, and convergence of average water depth towards an invariant distribution, which can be represented respectively as:

\begin{equation}
	\label{eq:2.3}
	h_i \geq 0, \quad\|V\|_1=\|V_0\|_1,\quad \frac{dW}{dt} \leq 0, \quad \lim_{t \to \infty} h = \bar{h},
\end{equation}
where $V=[V_1,V_2,\ ...,V_n]$ is the vector of cellular water volumes, $\|*\|_1$ denotes the 1-norm of a vector, $W$ signifies the total gravitational potential energy of water in D, $h=[h_1,h_2,\ ...,h_n]$ represents average water depth vector and $\bar{h}$ is the final invariant distribution corresponding to $h$. Note that the initial water velocity distribution is assumed to be zero and the reference point for gravitational potential energy is defined at $Z=0$.

To ensure non-negativity and conservation conditions as stated in \ref{eq:2.3}, Lemma 2.1 provides an evident sufficient condition for constraining rules matrix A. From a physical perspective, this implies that if inflow and outflow have equal magnitudes but opposite directions at a given boundary, then total water volume remains conserved; furthermore, if outflow ceases when water depth reaches zero, negative values for water depth are prevented.

\textbf{Lemma 2.1:} Suppose $\left|a_{ij}\right|<+\infty$. If $A=-A^T$, then the solution of dynamical system \ref{eq:2.2} remains conserved, that is $\|V(t)\|_1=\|V_0\|_1$. If $a_{ij}\geq 0$ when $h_i \le 0$, then the solution of the dynamical system \ref{eq:2.2} is non-negative. 

To meet the condition in \ref{eq:2.3} of non-increase in total gravitational potential energy, it is necessary to calculate the total gravitational potential energy function $W(t)$ of water in region $D$. The total gravitational potential energy function of water $w_i$ in ${Cell}_i$ can be written as:

\begin{equation}
		\label{eq:2.4}
w_i=\iint_{{cell}_i}\left[\rho g h(x, y) z(x, y)+\frac{1}{2} \rho g h(x, y)^2\right] d x d y
\end{equation}

To make the integral in \ref{eq:2.4} computable, assumption A1 is introduced, which is not unique.

\textbf{Assumption A1:} $z\left(x,y\right)$ and $h\left(x,y\right)$ is flat (piece-wise constant) in the cell. 

Based on assumption A1, a sufficient condition is given by Lemma 2.2 to ensure the non-increase in total gravitational potential energy of water in D.

\textbf{Lemma 2.2:} Given rules matrix A satisfying the Lemma 2.1, if $H_i\geq H_j$ when $a_{ij}\le 0$, then the solution of dynamical system \ref{eq:2.2} meets non-increase in total gravitational potential energy of water in $D$. 

\begin{proof}[Proof of Lemma 2.2]
	
From the function \ref{eq:2.4} and the assumption A1, Wt can be derived as

\begin{equation}
W\left(t\right)=\sum_{i=1}^{n}  w_i=\sum_{i=1}^{n} \rho g\left(z_ih_i+1/2h_i^2\right)S_i \nonumber
\end{equation}

Take the derivative of time on both sides of the equal sign: 

\begin{equation}
\frac{dW\left(t\right)}{dt}=\rho g\sum_{i=1}^{n}  S_iH_i\frac{dh_i}{dt}=\rho g\sum_{i=1}^{n}  A_iIS_iH_i/S_i=\rho g\sum_{i=1}^{n} \sum_{j=1}^{n}  H_ia_{ij} \nonumber
\end{equation}

According to $A=-A^T$ from Lemma 2.1, we get $H_ia_{ij}+H_ja_{ji}=\left(H_i-H_j\right)a_{ij}$.
Let $L=\left\{(i,j)|a_{ij}<0\ or\ (a_{ij}=0\ and\ i<j)\right\}$, then $dW\left(t\right)/dt$ can be rewrite as the following form:

\begin{equation}
	\label{eq:2.5}
	\frac{d W(t)}{d t}=\rho g \sum_{(i, j) \in L}\left(H_i-H_j\right) a_{i j}
\end{equation}

It is clear that if $H_i\geq H_j$ when $a_{ij}\le 0$ for all $i$ and $j$, then $dW\left(t\right)/dt\le 0$ holds. 

By combining Lemma 2.1 and Lemma 2.2, we propose the following theorem, which ensures that the solution of the dynamical system \ref{eq:2.2} could satisfy all physical conditions in \ref{eq:2.3}. 

\end{proof}

\begin{theorem}\label{thm1}
	If the following five conditions for the rules matrix hold, then the four physical conditions in \ref{eq:2.3} for the solution of the dynamical system \ref{eq:2.2} all hold. 
	
	\begin{enumerate}[a.]
		\item $A=-A^T$, 
		
		\item $\left|a_{ij}\right|<+\infty$,
		
		\item if $H_i>H_j$ and $h_i\le 0$, then $a_{ij}=0$,
		
		\item if $H_i>H_j$ and $h_i>0$, then $a_{ij}<0$, 
		
		\item if $H_i=H_j$, then $a_{ij}=0$.
		
	\end{enumerate}
	
\end{theorem}

\begin{proof}[Proof of Theorem~{\rm\ref{thm1}}]
	Based on the conditions of the theorem \ref{thm1}, it is easy to derive Lemma 2.1 and Lemma 2.2. Next, only prove the conclusion $\lim_{t \to \infty} h = \bar{h}$ in \ref{eq:2.3}. 
	
	First step, let $S= \{h | if \quad H_i\geq H_j,then  \quad h_i=0  \quad or  \quad H_i=H_j\}$ , then $h\in S$ if and only if $dW(h)/dt=0$. 
	
	(1) Proof of necessity. If $h\in S$, there is $a_{ij}=0$ for all $i$ and $j$, so $dW(h)/dt=0$ holds according to formula \ref{eq:2.5}. 
	
	(2) Proof of sufficiency. If $dW(h)/dt=0$, then $a_{ij}<0$ is impossible according to the theorem conditions. Because, under the condition $a_{ij}<0$, $H_i=H_j$ contradicts the condition (e), $H_i>H_j$ and $h_i=0$ contradicts the condition (c), $H_i<H_j$ and $h_j>0$ contradicts the condition (d), $H_i<H_j$ and $h_j=0$ contradicts the condition (c), $H_i>H_j$ and $h_i>0$ contradicts the $dW(h)/dt=0$. So, if $dW(h)/dt=0$, then $a_{ij}=0$. Under the condition $a_{ij}=0$, $H_i>H_j$ and $h_i>0$ contradicts the condition (d), so there is $H_i \le H_j$ or $h_i=0$. Due to $a_{ji}=0$, there is $ H_i\geq H_j$ or $h_j=0$. So, if $dW(h)/dt=0$, there is $h_i=0$ or $H_i=H_j$, which means $h\in S$. 
	
	Second step, if $h_0\in S$, $\lim_{t \to \infty} h = \bar{h}$ holds obviously. Next suppose $h_0\notin S$. For $\forall h\notin S$, there is $dW(h)/dt<0$. And $W(t)$ has a lower bound, so $\lim_{t \to \infty} dW(h)/dt=0$ holds. So, there exists a unique $\bar{h}\in S$ (Suppose the solution of the dynamic system \ref{eq:2.2} exists and is unique), so that $\lim_{t \to \infty} h = \bar{h}$ holds.  Moreover, there is $d\bar{h}/dt=0$, which means the $\bar{h}$  is a fixed point of the dynamic system \ref{eq:2.2}. If take $W(t)$ as Lyapunov Function, it is obvious that the fixed point $\bar{h}$ is Lyapunov-stable \cite{wiggins2010}. 
	
\end{proof}

The above describes the theoretical basis of the urban flood dynamical system model. Next, the urban flood process modeling is based on this theory, which mainly includes the construction of the urban cell domain and the selection of the flow rules matrix. The urban cell domain uses the diversity of cell types to reflect the diversity of urban ground use types and absorb as much data as possible within the city. The flow rules matrix reflects the complexity of urban flow situations with the diversity of water movement or distribution rules.

Construction of urban cell domain: diversity of cell types. Urban ground cover has an important impact on flood generation, runoff and confluence processes. The ground conditions in urban areas are complex and the types of ground uses are diverse, such as buildings, roads, grasslands, open spaces, water bodies, underground spaces, mountains, etc. Different types of ground use have different processes of flood generation, runoff and confluence. For example, grassland has high water permeability, and the infiltration and interception process is obvious. Road water permeability is low, and the runoff process is significant, and the infiltration process can be ignored. At present, when most flood models are applied to complex urban systems, they describe different ground utilization types, usually using different ground friction coefficients, or using obstacles \cite{luo2022}. In the UFDSM, cells can naturally be set to different types, so that various types of ground use in the city can be easily described, such as building cells and road cells, as shown in Fig\ref{fig1}(a).Different cell types have different attributes, flow generation rules, and water interaction rules. In order to uniformly describe the urban flood process in the cell domain, the UFDSM is also equipped with boundary cells and rain cells. Boundary cells are used to carry various inflow and outflow boundary conditions, and rain cells are used to play the role of rainfall and can express spatially differentiated rainfall. In this paper, boundary cells and rain cells are called virtual cells, and other cells are called real cells. Urban hydrology and water conservancy facilities play an important role in mitigating flood processes, such as urban drainage pipe networks, pumping stations and artificial lake. The urban flood dynamic system model assigns the real cell type whether it has drainage outlet attributes to reflect the drainage effect, and whether it has other source term attributes to increase the degree of freedom and adapt to the real situation. In addition, the algorithm architecture of the UFDSM can easily be expanded to more customized cell types. 

Construction of flow rules matrix: water movement or distribution rule diversity. The urban surface conditions are intricate, and the flow phenomena exhibit a diverse range of characteristics. To enhance the portrayal of complex urban systems, we couple the urban flood dynamic system model with multiple sets of water flux calculation rules. This allows for each edge to be designed with flow rules that are adapted to actual surface conditions, providing great flexibility and powerful expression ability. Common water flux calculation rules include Manning's formula, weir flow formula, hole flow formula, inertial simplified formula of Saint-Venant's equation, dynamic wave equation, etc. It is important to note that when constructing the rules matrix based on the aforementioned calculation rules, it is necessary to satisfy the conditions stated in Theorem \ref{thm1}. Therefore, adjustments may need to be made to these rules if required.

The construction process of the urban flood dynamic system model for the simulated urban area (with as simple a boundary as possible) is as follows (see Fig \ref{fig1}):

\begin{figure*}[t]
\centerline{\includegraphics[width=0.7\textwidth]{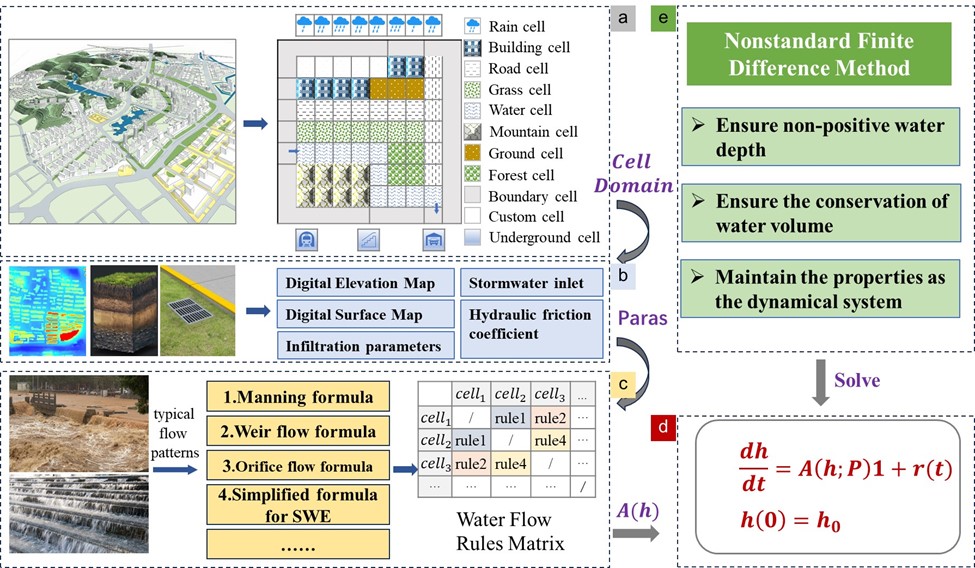}}
\caption{The simulation framework of the urban flood dynamical system model. (a) Construction of the cell domain, (b) Setting cell properties and parameters, (c) Construction of the flow rules matrix,  (d) Urban flood dynamical system model equations, (e) Non-standard finite difference algorithm, used to solve the constructed urban flood dynamical system equations.\label{fig1}}
\end{figure*}

 (1) Divide the city into various types of cells based on the urban ground coverage map (such as a city street map), such as building cells, road cells, underground shopping malls or subway cells, etc., while setting up boundary cells, rainfall cells, and special custom cells. The shape of the cell boundary can be irregular, but regular boundaries can simplify the calculation process. After division, a cell domain $CD={cell_1,cell_2,\ldots,cell_n}$ and a type set $type={building\ cell,road\ cell,\ldots}$ are formed.  For any $cell_i\in CD$, $cell_i$ has a type attribute, that is, $type\left(cell_i\right)\in type$.
 
 (2) Set the elevation of each cell based on the Digital Elevation Model (DEM) or\\and Digital Surface Model (DSM), set the infiltration parameters and hydraulic friction coefficients of the cells based on the urban underlying surface properties, set the drainage source term based on whether there is a drainage outlet, and set other source terms based on actual conditions. After the settings are completed, a cell parameters set $P= \{parameters (cell) |cell\in CD\}$ is formed.
 
(3) Based on the actual flow phenomenon between neighboring cells, construct flow rules matrix $A(h,P)$ and make the matrix satisfy the conditions required by Theorem \ref{thm1}.

(4) Establish the dynamical system model equation, $\frac{dh}{dt}=A\left(h,p\right)\mathbf{1}+r\left(t\right)$, and specify its initial condition $h\left(0\right)=h_0$.

(5) Based on the proposed non-standard finite difference algorithm to solve the dynamical system, the algorithm can maintain the non-negativity of water depth, the conservation of  water volume, and remain consistent with the evolutionary properties of the dynamical system.

\section{Construction of the numerical method}\label{sec3} 

 Once constructing rules matrix A based on the idea of Cellular Automata, a dynamical system is established. Consequently, how to numerically solve the dynamical system is no longer limited by the original updating strategy of Cellular Automata \cite{nagel1992,wolfram1984}. In the subsequent section, we also suppose that the dynamical system \ref{eq:2.2} is well-posed \cite{wiggins2010}. To ensure that the numerical solution algorithm for dynamical system \ref{eq:2.2} adheres as closely as possible to physical constraints in \ref{eq:2.3}, a class of the non-standard finite difference methods is selected here, which can preserve the water depth positivity of the evolution process, conservation of water volume, the final invariant distribution and the local property of the fixed point \cite{wood2015}.
 
 A finite difference scheme approximating dynamical system \ref{eq:2.2} can be written as\cite{wood22015}: 
 
\begin{equation}
 	\label{eq:3.1}
 D_\tau\left(h_i^k\right)=F_\tau\left(A \mathbf{1} ; h_i^k\right),
\end{equation}
where $\left.D_\tau\left(h_i^k\right)\approx\frac{dh_i}{dt}\right|_{t=t_k}, h_i^k\approx h_i\left(t_k\right), F_\tau\left(A\mathbf{1};h_i^k\right) approximates A\mathbf{1}/S_i$ in the dynamical system \ref{eq:2.2}, $t_k=t_0+k\tau$, where $h>0$.

With reference to this research \cite{wood2015}, the numerical method based on the non-standard finite difference frame approximating the dynamical system \ref{eq:2.2} is established as follows:

 \begin{equation}
	\label{eq:3.2}
	\frac{h_i^{k+1}-h_i^k}{\phi(\tau)}=\left\{\begin{array}{cl}
		A_i^k \mathbf{1} / S_i, & if A_i^k \mathbf{1} \geq 0 \\
		\frac{h_i^{k+1}}{h_i^k} A_i^k \mathbf{1} / S_i, & if A_i^k \mathbf{1}<0
	\end{array},\right.
\end{equation}
where $\mathrm{\ }A_i^k=\ A_i\left(h^k\right)$ and $A_i=\left[a_{i1},a_{i2},...,a_{in}\right]$.

Rewrite it in explicit form:

\begin{equation}
	h_i^{k+1}=\left\{\begin{array}{cl}
		h_i^k+\phi(\tau) A_i^k 1 / S_i, & \text { if } A_i^k I \geq 0 \\
		h_i^k+\frac{h_i^k \phi(\tau) A_i^k 1 / S_i}{h_i^k-\phi(\tau) A_i^k 1 / S_i}, & \text { if } A_i^k I<0
	\end{array}\right. \nonumber
\end{equation}

Daniel T. Wood has demonstrated that the numerical method \ref{eq:3.2} is positive and has elementary stable properties under certain conditions \cite{wood2015}. The definitions of positive and elementary stable are quoted here. 

The finite difference method \ref{eq:3.1} is called positive, if, for any value of the step size $\tau$ and initial location $h_0 \in R_+^n$, the iterated sequence solution  $\{h^k\}_{k=1}^{\infty}$ of the numerical method \ref{eq:3.2} has $h \in R_+^n$ for any $k$, where $R_+^n=\left\{(x_1,x_2,...,x_n)^T|x_i\geq0,i=1,2,...,n\right\}$. 

The finite difference method \ref{eq:3.1} is referred to as elementary stable if, for any value of the step size $\tau$, its fixed point $\bar h$ are the same as the equilibria of the differential system \ref{eq:2.2} and the local stability properties of each $\bar h$ are the same for both the differential system and the difference method.

Let $\varphi(\tau)$ be a real-valued function which satisfies $0<\varphi\left(\tau\right)<1$  for all $\tau>0$ and $\varphi\left(\tau\right)=\tau+O(\tau)$. Let $\varphi(h)=\phi(q\tau)/q$. if taking $q>Q$, the numerical method \ref{eq:3.2} can be elementary stable, where $Q$ satisfies:

\begin{equation}
	\label{eq:3.3}
	Q \ge \max \{\frac{| \lambda | ^2}{2|Re \lambda |}\},
\end{equation}
where $\lambda$ is the eigenvalue of Jacobian matrix of the dynamical system \ref{eq:2.2}.

More detailed definitions and theories of non-standard finite difference methods can be referred to those articles \cite{anguelov2001, lubuma2005, mickens1994}.

The inadequacy of this format in ensuring water volume conservation is evident, as the computed flux of upstream and downstream cells on both sides of the same boundary exhibits inconsistency. To ensure water volume conservation, an upwind splitting approach is employed for each boundary, where the outflow from upstream is considered as the inflow to downstream.

Split the rules matrix $A^k$ into $A^{k+}$ and $A^{k-}$, that is $A^k=A^{k+}+A^{k-}$, where $A^{k+} $ representing inflow flux $(a_{ij}^{k+}>0)$ and $A^{k-}$ representing the outflow flux $(a_{ij}^{k-}\le0)$. Taking splitting calculation, firstly, the outflow flux (non-positive) during $[t_k,t_k+\tau]$ is calculated as follows based on upwind idea:

\begin{equation}
	f_i^{k\mathrm{\ (out)\ }}=\frac{\phi\left(\tau\right)A_i^{k-}\mathbf{1}/S_i}{h_i^k-\phi\left(\tau\right)A_i^{k-}\mathbf{1}/S_i}h_i^k=\frac{\phi\left(\tau\right)\sum_{j=1}^{n}  a_{ij}^{k-}/S_i}{h_i^k-\phi\left(\tau\right)A_i^{k-}\mathbf{1}/S_i}h_i^k  \nonumber
\end{equation}

Let $G^k=(g_{ij})^k=\frac{a_{ij}^{k-}/S_i}{h_i^k-\phi\left(\tau\right)A_i^{k-}I/S_i}h_i^k$, then $f_i^{k\mathrm{\ (out)\ }}=\phi\left(\tau\right)G_i^k\mathbf{1}$

Next, the inflow flux is obtained from the just calculated outflow flux:

\begin{equation}
	f_i^{k\left(in\right)\mathrm{\ } }=-\phi\left(\tau\right){{(G}^{k\top})}_i\mathbf{1}, \nonumber
\end{equation}
where $G^{k\top}$ is the transpose of $G^k$.

Finally, the next step water depth can be updated in the following formula:

\begin{equation}
	\label{eq:3.4}
	h^{k+1}=h^k+\phi(\tau)(G^k-G^{k\top})\mathbf{1}
\end{equation}

Let $U^k=(u_{ij})^k=\frac{a_{ij}^{k-}/S_i}{h_i^k-\phi\left(\tau\right)A_i^{k-}I/S_i}$, and $\Lambda^k$ be the diagonal matrix spanned by the vector$ U^k\ \mathbf{1}$. Then, the formula \ref{eq:3.4} could be rewritten as:

\begin{equation}
	h^{k+1}=h^k+\phi\left(\tau\right)\left(\Lambda^k-U^{k\top}\right)h^k=\left[E+\phi\left(\tau\right)\left(\Lambda^k-U^{k\top}\right)\right]h^k  \nonumber
\end{equation}

An analysis is conducted on the local truncation error, convergence, and compatibility of the numerical method \ref{eq:3.4}. We have reached the following conclusion. Given the initial water depth $h^0 \in R_+^n$, if the area $S$ of every cell in domain $D$ is the same, the numerical method \ref{eq:3.4} is convergent, and its local truncation error is $O (\tau^2)$. Prove as follows.

The local truncation error $R_k$ can be presented as:

\begin{equation}
	R_i^k=h_i\left(t_{k+1}\right)-h_i^{k+1}=h_i\left(t_{k+1}\right)-h_i\left(t_k\right)-\phi\left(\tau\right)\left(G_i^k-G_i^{kT}\right)\mathbf{1} \nonumber
\end{equation}

Since $A=-A^\top$ and $A^k=A^{k+}+A^{k-}$, there is $a_{ij}^+=-a_{ji}^-$. Suppose that rules matrix $A$ is sufficiently smooth, and take $\phi\left(\tau\right)=\tau+O\left(\tau^2\right)$ based on the NSFD frame. Then, according to Taylor formula, the $R_k$ can be expanded and simplified as follow:

\begin{equation}
	R_i^k=\tau\left[A_i/S_i-G_i^k+(G^{kT})_i\right]\mathbf{1}+\left[1-G_i^k\mathbf{1}+\left(G^{kT})_i\mathbf{1}\right)\right]O(\tau^2),
	\nonumber
\end{equation}
where $A_i/S_i-G_i^k+(G^{kT})_i$ can be expanded as:

\begin{equation}
	A_i/S_i-G_i^k+(G^{kT})_i=\frac{h_j^ka_{ji}^-(S_i-S_j)/(S_iS_j)}{h_j^k-\phi(\tau)A_j^{k-}\mathbf{1}/S_j}-\phi(\tau)\left[\frac{a_{ij}^-A_i^{k-}\mathbf{1}/S_i^2}{h_i^k-\phi(\tau)A_i^{k-}\mathbf{1}/S_i}-\frac{a_{ji}^-A_j^{k-}\mathbf{1}/(S_iS_j)}{h_j^k-\phi(\tau)A_j^{k-}\mathbf{1}/S_j}\right]
	\nonumber
\end{equation}

If $S_i=S_j$ for any $i$ and $j$, we get

\begin{equation}
	R_i^k=\tau\phi(\tau)\sum_{j\in N_i} \left[\frac{-a_{ij}^-A_i^{k-}\mathbf{1}/S_i^2}{h_i^k-\phi(\tau)A_i^{k-}\mathbf{1}/S_i}+\frac{a_{ji}^-A_j^{k-}\mathbf{1}/(S_iS_j)}{h_j^k-\phi(\tau)A_j^{k-}\mathbf{1}/S_j}\right]+\left[1-G_i^k\mathbf{1}+\left(G^{kT})_i\mathbf{1}\right)\right]O(\tau^2)
	\nonumber
\end{equation}

If $h_i^k=0$, there is $a_{ij}^-=0$, and if $h_j^k=0$, there is $a_{ji}^-=0$. So, we get that $R_k=O\left(\tau^2\right)$, which implies the numerical method \ref{eq:3.4} is compatible and convergent.

Next, conduct an analysis on the stability of the numerical method \ref{eq:3.4}. Firstly, select a model equation. Consider a rectangular region $D$ consisting of square cells, with an initial water depth distribution $h\left(0\right)>0$ and a flat elevation distribution. Define the evolution rules as follows:

\begin{equation}
	\label{eq:3.5}
	\frac{d h_i}{d t}=-\frac{c l}{s} \sum_{j \in N_i}\left(h_i-h_j\right), \quad \frac{d h}{d t}=-\frac{c l}{s} B h,
\end{equation}
where $s$ is the cell area, $l$ is the boundary length of the cell, $c$ is the outflow coefficient, and the right-hand equation is the vector form where $B$ is denoted by:

\begin{equation}
	b_{i j}=\left\{\begin{array}{cc}
		-1, & j \in N_i \\
		0, & j \notin N_i \\
		\left|N_i\right|, & j=i
	\end{array},\right.
	\nonumber
\end{equation}
where $\left|N_i\right|$ represents the total number of neighborhoods of ${Cell}_i$.

For the model equation \ref{eq:3.5}, based on the numerical method \ref{eq:3.4}, define the $w_i$ and matrix $D$ as

\begin{equation}
	\begin{gathered}
		w_i=\frac{h_i^k}{h_i^k+\phi(\tau) c l \sum_{j \in N_i}} \max \left\{h_i-h_j, 0\right\} / s \\
		D=d_{i j}= \begin{cases}-w_i, & j \in N_i \& h_i \geq h_j \\
			-w_j, & j \in N_i \text \& h_i<h_j \\
			-\sum_{j \in N_i} d_{i j}, & i=j\end{cases}
	\end{gathered}
	\nonumber
\end{equation}

Obviously, $w_i\le 1$, and $D$ is a symmetric matrix. Then, the numerical equation \ref{eq:3.4} corresponding to the model equation \ref{eq:3.5} can be written as $h^{k+1}=h^k-\phi\left(\tau\right)clD^kh^k/s$. The corresponding error form is

\begin{equation}
	\label{eq:3.6}
	e^{k+1}=e^{k}-\phi(\tau)cl(\bar D \bar h-D^k h^k)/s,
\end{equation}
where the $\bar{h}$ is$ h^k$ with errors, $D^k=D\left(h^k\right), \bar{D}=D(\ \bar{h})$. Due to $\bar{D}\approx D^k$, the error equation \ref{eq:3.6} can be rewritten as $e^{k+1}=\left(E-\phi\left(\tau\right)clD^k/s\right)e^k$, where $E-\phi\left(\tau\right)clD^k/s$ is a symmetric matrix. To meet stability requirements, there is spectral radius $\rho\left(E-\phi\left(\tau\right)clD^k/s\right)\le1$. Then, we get a stability condition:

\begin{equation}
	\label{eq:3.7}
	\phi(\tau)\le\frac{2s}{cl\rho(B)}\le \frac{2s}{cl\rho(D^k)}.
\end{equation}

Since $\rho\left(B\right)$ is less than arbitrary induced matrix norm and $\|B\|_1=8$, then $\phi\left(\tau\right)$ can be estimated as $\phi\left(\tau\right)\le0.25s/cl$.

In addition to numerical stability, the numerical equation \ref{eq:3.4} should also be elementary stable. The gradient of model equation \ref{eq:3.2} with respect to h at the fixed point is $-clB/s$, which is a real symmetric matrix with spectral radius $cl\rho\left(B\right)/s$. According to the conditions requirements of elementary stable (see Formula \ref{eq:3.3}), $q$ should meet

\begin{equation}
	q>Q\ge \max\{\frac{|\lambda|^2}{2|Re(\lambda)|}\}=\frac{cl\rho(B)}{2s}.
\nonumber
\end{equation}   

Obviously, if $q$ satisfies this condition, $\phi(\tau)$ naturally satisfies the condition \ref{eq:3.7}. So, if $q>Q$, the numerical method \ref{eq:3.4} has both numerical stability and elementary stable when solving the model equation \ref{eq:3.5}.

\section{Numerical experiments}\label{sec4}

To validate the proposed numerical method \ref{eq:3.4} for computing urban flood dynamical systems \ref{eq:2.1}, two numerical experiments were conducted in this chapter. The first numerical experiment is based on the model equation \ref{eq:3.5}, which has analytical solutions. In the second numerical experiment, the model and the numerical method are applied to the flood simulation in small urban areas, with a comparison made against the flood simulation software HEC-RAC \cite{suriya2012}.

\subsection{Experiment with analytical solutions}

An 80m $\times$ 80m square area is regularly divided into 8 $\times$ 8 cells with cell size 10m $\times$ 10m  (as shown in Fig \ref{fig2}). The upper-left corner cell number is (0, 0), and the lower-right corner cell number is (7, 7). Based on model equation \ref{eq:3.5}, the evolution rules of water depth is defined as:

\begin{equation}
	\label{eq:4.1}
	\frac{d h}{d t}=\frac{d H}{d t}=\left\{\begin{array}{cl}
		0 & , h=0 \& B H>0 \\
		\frac{-c l}{s} B H & , \text { else }
	\end{array},\right.
\end{equation}
where $c=0.5m^2/s$, $s=100m^2$, $l=10m$, $H=z+h$. The elevation distribution z and initial water depth distribution $h_0$ in the region are deliberately designed (see Fig \ref{fig2}(a) and Fig \ref{fig2}(b)) so that the analytical solution of equation \ref{eq:4.1} can be obtained. As time progresses, the depth of the cell (0, 0) will decrease to 0 near 7.7s, after which it remains constant at zero. The water depths of other cells remain above zero throughout the entire duration, with a final convergence towards a level of 0.9m.

The analytical solution for this example is presented as follows:

\begin{equation}
	\nonumber
	H(t)=\left\{\begin{array}{cl}
		Q E(\Lambda) Q^T H_0 & , h_0>0 \\
		P E\left(\Lambda_{\mathrm{p}}\right) P^T H_{t_m} & , h_0=0
	\end{array},\right. 
\end{equation}
where $\mathrm{\Lambda}$ is the eigenvalue matrix of $B$, $Q$ is the corresponding orthogonalized eigenvector matrix,

\begin{equation}
	\label{eq:4.2}
	E(\Lambda)_{i j}=\left\{\begin{array}{cc}
		\exp \left(-\frac{c l}{s} \lambda_i t\right) & , i=j \\
		0 & , i \neq j
	\end{array},\right.
\end{equation}
where $\lambda_i$ is the $i-th$ diagonal element of the eigenvalue matrix $\mathrm{\Lambda}$, $t_m$ is the moment when the water depth of  the cell (0, 0) is zero for the first time, $\mathrm{\Lambda}_p$ is the eigenvalue matrix of the adjacency matrix $B_p$, $B_p$ is the new adjacency matrix of the cell domain after the removal of the  cell (0, 0), $P$ is the eigenvector matrix of $B_p$ corresponding to $\mathrm{\Lambda}_p$, $E\left(\Lambda_p\right)$ is synonymous with the formula \ref{eq:4.2}, $H_{t_m}$ is the cellular water level at $t_m$, excluding the cell (0, 0).

In addition to the exact solution, the 4-level, 4-order explicit Runge-Kutta method \cite{shampine1997} is used to solve the dynamic system \ref{eq:4.1} and compared with the numerical method \ref{eq:3.4}. Because the 4-level, 4-order explicit Runge-Kutta method may yield negative water depth, we set the water depth of a cell to 0 if the water depth of a cell is negative. Set the time step $\tau=1s$. In the numerical method \ref{eq:3.4}, we utilize $\phi(\tau) = [1 - \exp(-q\tau)]/q$ and $q =0.3$ satisfying the condition \ref{eq:3.7}. In Fig \ref{fig2}(c), four cells located on the diagonal line are selected as comparison points, numbered as (0,0), (1,1), (4,4), and (7,7). As depicted in the figure, all three results are consistent with each other. On the comparison cells, the maximum error $\|H-H_{NSFD}\|_{\infty}$ between the water depth calculated by the numerical algorithm \ref{eq:3.1} and the exact solution are found to be 0.0014m, 0.0005m, 0.0011m, 0.0006m respectively. By employing modified Runge-Kutta method, the maximum error $\|H-H_{RK}\|_{\infty}$ between the calculated water depth and exact solution are found to be $4.0\times10^{-7}m, 1.1\times{10}^{-5}m, 1.1\times{10}^{-6}m,\ 6.4\times{10}^{-7}m$ respectively. The aforementioned magnitude of errors arises due to difference in numerical precision levels; while numerical algorithm \ref{eq:3.1} has first-order precision, the Runge-Kutta method exhibits 4-order precision. However, in actual complex urban flooding flow scenarios, a large number of cellular water depths may change to zero, resulting in significant uncertainties associated with this modified Runge-Kutta algorithm. It is important to note that numerical algorithm \ref{eq:3.1} naturally maintains positive properties of solutions.

\begin{figure}[htbp!]
	\centering  
	\begin{subfigure}{0.48\textwidth}  
		\centering
		\includegraphics[height=15pc]{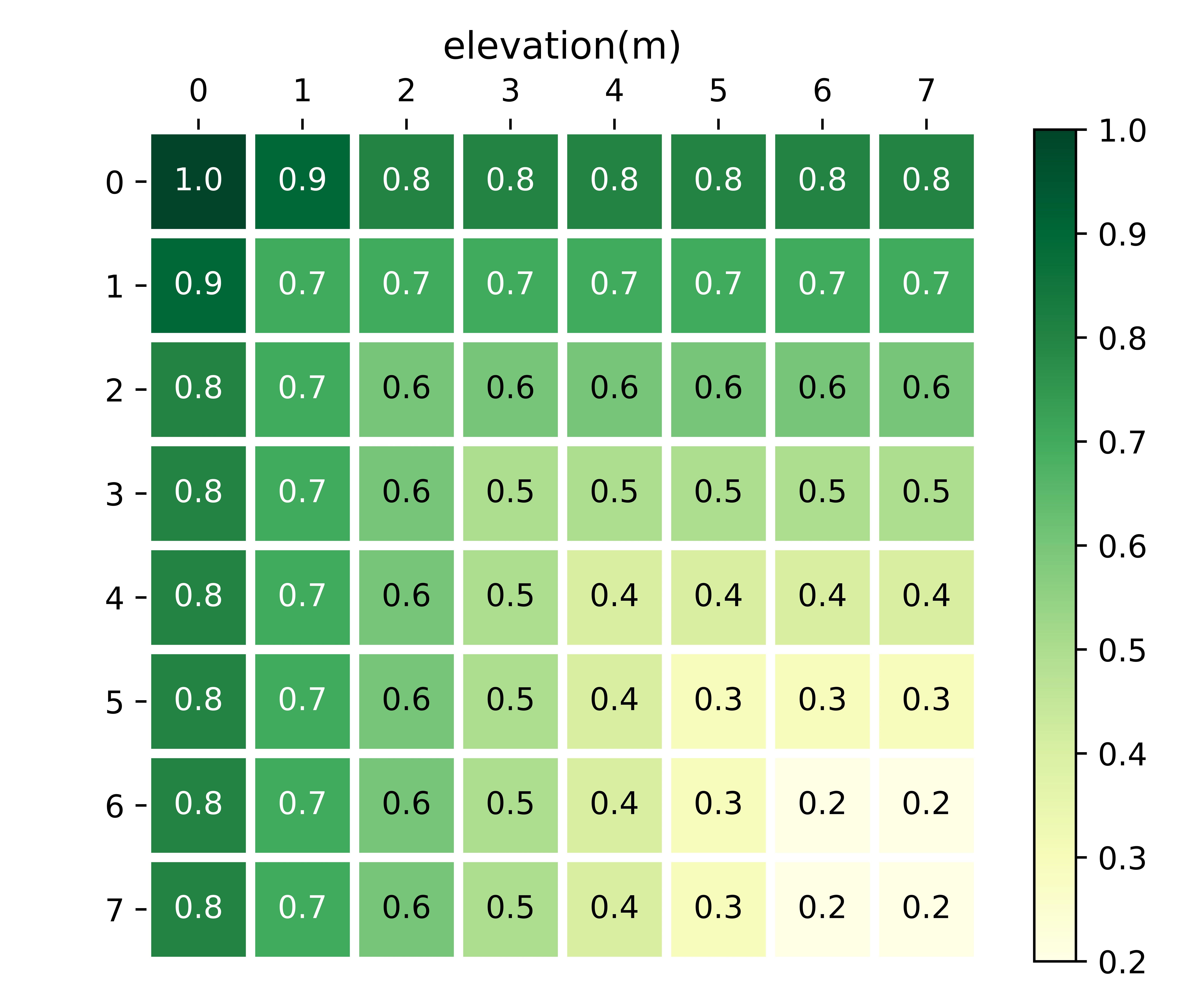}
		\subcaption{ }
	\end{subfigure}
	\hfill  
	\begin{subfigure}{0.48\textwidth}
		\centering
		\includegraphics[height=15pc]{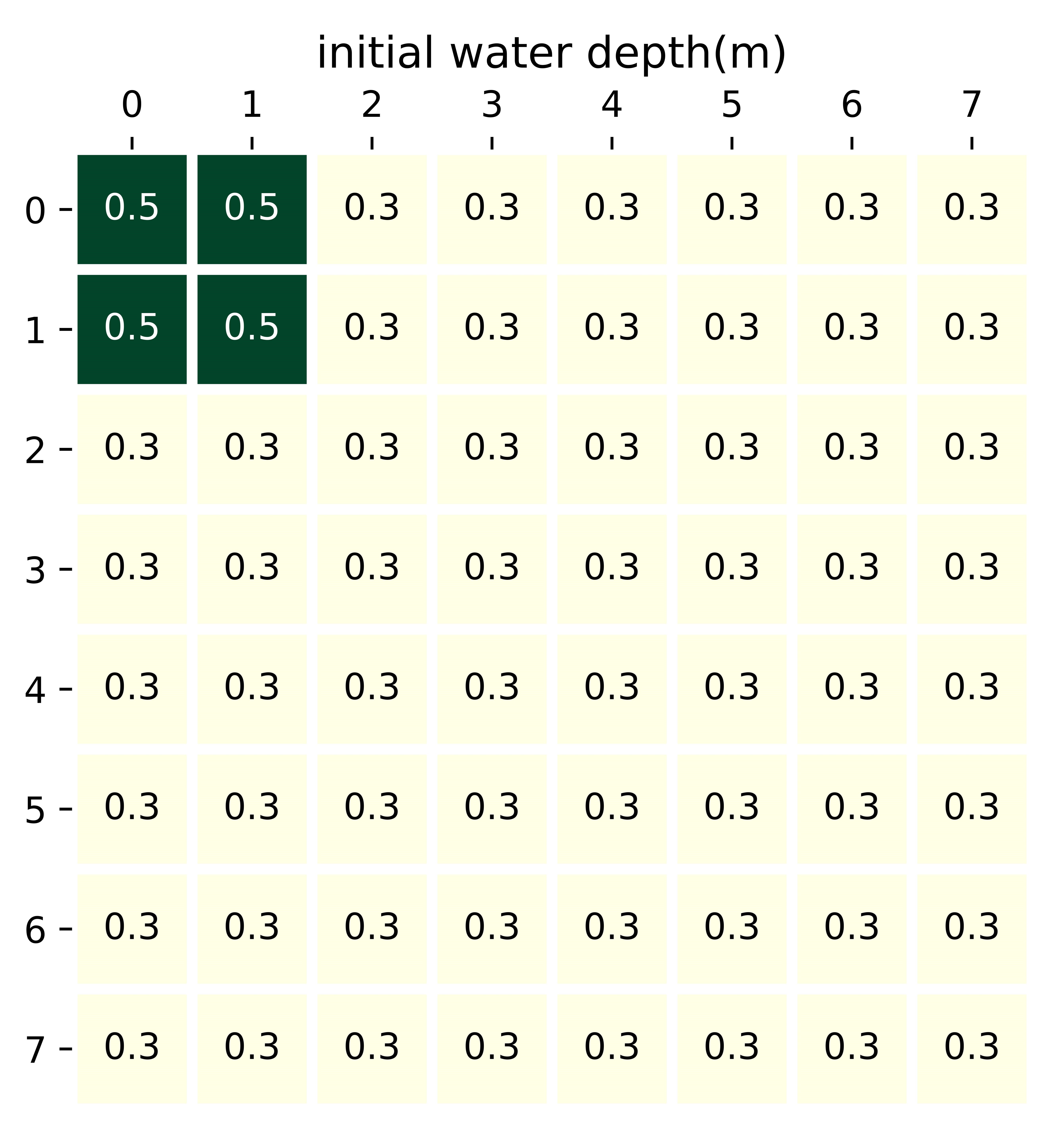}
		\subcaption{ }
	\end{subfigure}
	
	\begin{subfigure}{0.48\textwidth}
		\centering
		\includegraphics[height=15pc]{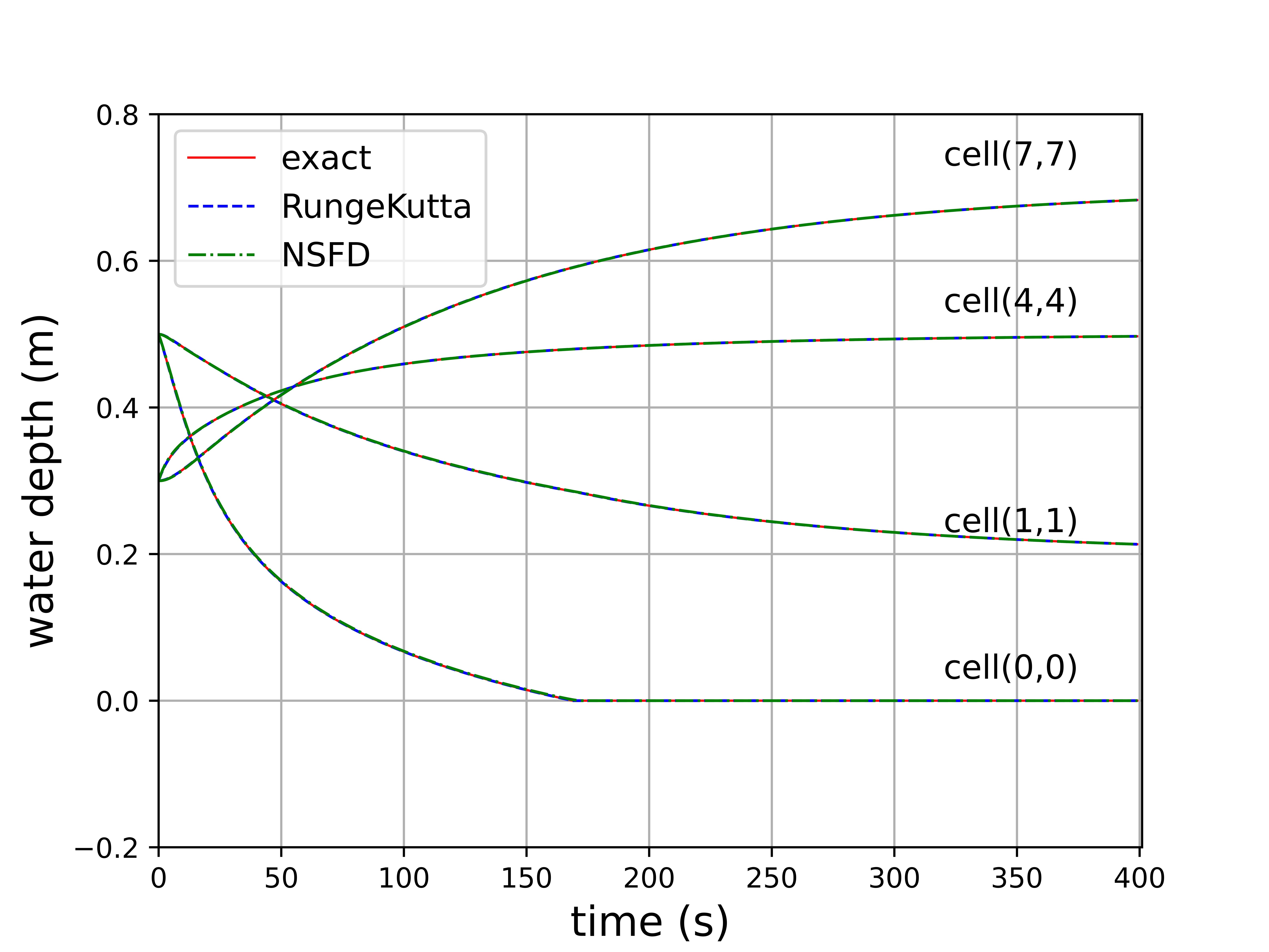}
		\subcaption{ }
	\end{subfigure}
	\hfill
	\begin{subfigure}{0.48\textwidth}
		\centering
		\includegraphics[height=15pc]{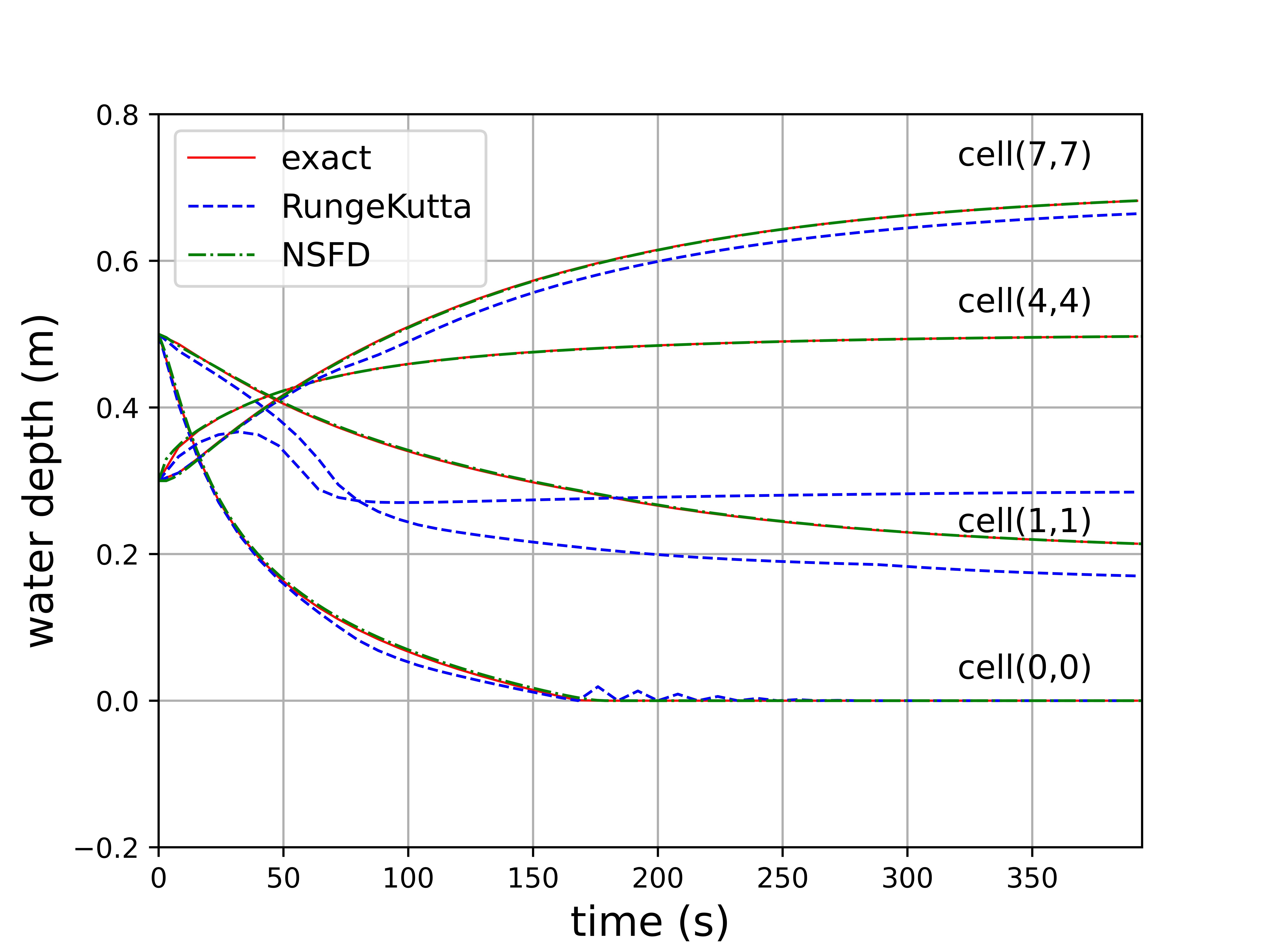}
		\subcaption{ }
	\end{subfigure}
	\caption{(a) elevation distribution, (b) initial water depth distribution, (c) the numerical solutions of the numerical method \ref{eq:3.4} and the Runge-Kutta method both with $\tau=1s$ for the dynamical system \ref{eq:4.1}, and the exact solution, (d) the numerical solutions of the numerical method \ref{eq:3.4} and the Runge-Kutta method both with $\tau=8s$ for the dynamical system \ref{eq:4.1}, and the exact solution. The calculation results for $\tau=8s$ are presented in Fig \ref{fig2}(d), demonstrating the close proximity of the numerical method \ref{eq:3.1} to the real solution. On the comparison cells, the maximum error values of the water depth calculated by the numerical algorithm \ref{eq:3.1} and the exact solution $\|H-H_{NSFD}\|_{\infty}$ are 0.0056m, 0.0024m, 0.0063m and 0.0024m, respectively. Because of the error propagation of the modified Runge-Kutta method, there is a large deviation from the exact solution. The maximum error values of the calculated water depth and the exact solution $\|H-H_{RK}\|_{\infty}$ are 0.019m, 0.097m, 0.21m and 0.020m, respectively. The above numerical experiments verify the conclusion of the section 3: if $q > Q$, the numerical method \ref{eq:3.1} has both numerical stability and elementary stability when solving the model equation \ref{eq:3.2}.  \label{fig2}}
\end{figure}

\subsection{Experiment in urban region}

The experimental simulation area encompasses a 400m $\times$ 512m urban region, with its Digital Elevation Map (DEM) depicted in Fig \ref{fig3}. The DEM exhibits a resolution of 2m $\times$ 2m. Flooding within the region is induced by a uniformly distributed rainfall event lasting for 9 minutes, resulting in a total rainfall depth of 58mm and an approximate total rainfall volume of 5,939.2$m^3$. During the initial 9 minutes, the rainfall intensity per minute follows the sequence (0,1,2,4,7,7,5,2,1,0) mm/min; linear interpolation is employed within each interval. As mentioned in the second part, rules matrix $A$ can be flexibly selected according to the actual urban environment under the condition of satisfying theorem \ref{thm1}. In practical urban flood simulation, commonly employed rules include Manning formula, weir flow formula, hole flow formula, simplified formula of Saint-Venant equation and so on.

The rules matrix $A$ is constructed based on Manning's formula in this experiment, and the outflow $f_{ij}$ from cell $i$ to cell $j$ is defined as:

\begin{equation}
	f_{ij}=\frac{\beta sign(sl_{ij})\alpha_{ij}R_{ij}^{3/2} |sl_{ij}|^{1/2}}{n_{ij}},
	\nonumber
\end{equation}
where $\alpha_{ij}$ represents the cross-sectional area of water flow between cell $i$ and cell $j$, $R_{ij}$ denotes the hydraulic radius between cell i and cell $j$, $sl_{ij}$ is the slope between cell $i$ and cell $j$, $n_{ij}$ is the Manning’s coefficient between cell $i$ and cell $j$, $sign(\times)$ is symbolic function, $\beta is$ control coefficient. $\alpha_{ij}$ is generally approximated in two ways: $l_{ij}(h_i+h_j)/2$ or $l_{ij}\left[(h_i+h_j)/2+sign(sl_{ij})(h_i-h_j)/2\right]$, where $l_{ij}$ denotes the length of two-cell common edge. The second approximation reflects the upwind properties of shallow water waves. In general, the length of the cell edge is much greater than the water depth, so $R_{ij}$ can be directly approximated by the water depth of the cell: $R_{ij}=(h_i+h_j)/2$ or $R_{ij}= (h_i+h_j)/2+sign(sl_{ij})(h_i-h_j)/2$. $sl_{ij}$ is generally approximated by the two-cell water level difference: $sl_{ij}=(H_i-H_j)/\Delta x_{ij}$, where $\Delta x_{ij}$ is the distance between two cell centers. $n_{ij}$ could take $n_{ij}=(n_i+n_j)/2$.

\begin{figure*}
	\centerline{\includegraphics[width=0.8\textwidth]{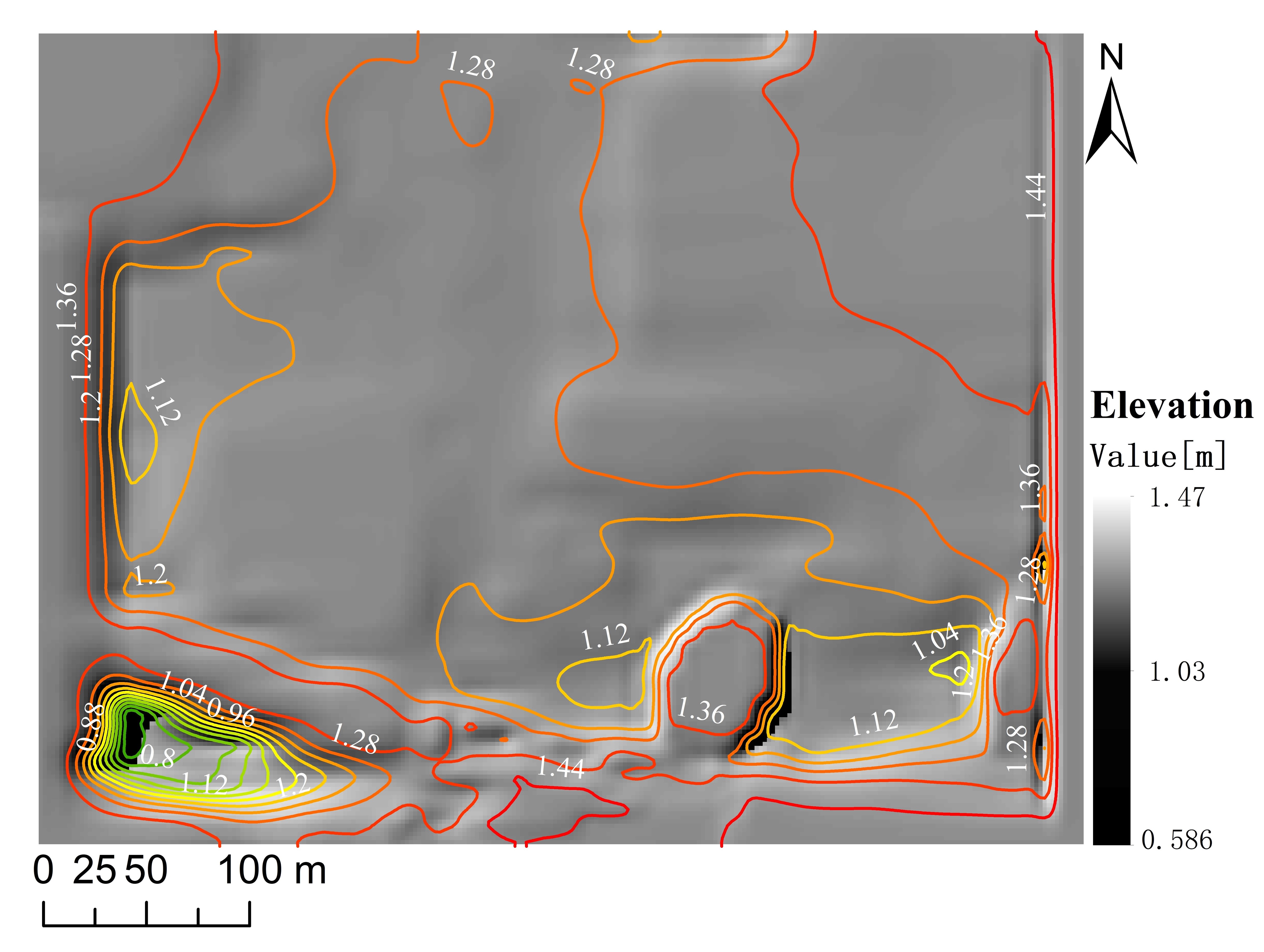}}
	\caption{The Contours and DEM shadow map of experimental simulation area \label{fig3}
	}
\end{figure*}

Based on the defined outflow $f_{ij}$, the water depth evolution process of the $i-th$ cell is constructed as follows:

\begin{equation}
	\label{eq:4.3}
	\frac{d h_i}{d t}=\left\{\begin{array}{cc}
		0 & , h_i \leq 0 \& \sum_{j \in N_i} f_{i j} \geq 0 \\
		-\frac{\sum_{j \in N_i} f_{i j}}{s_i} & , \text { else }
	\end{array},\right.
\end{equation}
where $s_i$ is the area of the $i-th$ cell.

The settings of this numerical experiment are: $\alpha_{ij}=l_{ij}(h_i+h_j)/2$, $R_{ij}=(h_i+h_j)/2$, $sl_{ij}=(H_i-H_j)/\Delta x_{ij}$, $n_{ij}=(n_i+n_j)/2$. It is easy to verify that the regular matrix $A$ formed by the equations \ref{eq:4.3} and the above settings satisfies all the conditions required by theorem \ref{thm1}. Then, given the initial water depth, the solution of the formed dynamic system naturally satisfies the four physical conditions in \ref{eq:2.3}. Let $F_i\left(h\right)=-\sum_{j\in N_i} f_{ij}/s_i and further F\left(h\right)=[F1,F2,…,FN]^{\top}$. Assume that $F\left(h\right)$ is differentiable at $h$, and each cell is square with the same size, then the Jacobian matrix of $F\left(h\right)$ is:

\begin{equation}
	\nabla F_{i j}=\frac{1}{2} k U_{i j}+\frac{5}{2} k V_{i j} ,
	\nonumber
\end{equation}
where

\begin{equation}
	k=\frac{l}{2^{5 / 2} sn \sqrt{\Delta x}} ,
	\nonumber
\end{equation}

\begin{equation}
	U_{i j}=\left\{\begin{array}{cc}
		\frac{\left(h_i+h_j\right)^{5 / 2}}{\sqrt{\left|H_i-H_j\right|}} & , j \in N_i \\
		-\sum_{j \in N_i} U_{i j} & , i=j \\
		0 & , \text { else }
	\end{array},\right. 
	\nonumber
\end{equation}

\begin{equation}
	V_{i j}=\left\{\begin{array}{cc}
		0 & , \text { else } \\
		\operatorname{sign}\left(H_j-H_i\right)\left(h_i+h_j\right)^{3 / 2} \sqrt{\left|H_i-H_j\right|} & , j \in N_i \text { and } i>j \\
		-V_{j i} & , j \in N_i \text { and } i<j \\
		\sum_{j \in N_i} V_{i j} & , i=j \\
		0 & , \text { else }
	\end{array}.\right.
	\nonumber
\end{equation}

Clearly, $U$ is a symmetric matrix and $V$ is an antisymmetric matrix. The denominator of $U_{ij}$ shows that $F\left(h\right)$ is not differentiable at $H_i=H_j$. When the water level tends to the horizontal, $V_{ij}$ tends to 0, and $U_{ij}$ tends to infinity. However, from the matrix structure of $U$ and $V$, the eigenvalues of the Jacobian matrix $\mathrm{\nabla F}$ are all less than or equal to 0. Therefore, the urban flood dynamical system is stable at the fixed point, which is consistent with the conclusion of the Theorem \ref{thm1}. It is worth noting that the depth vector only moves in the non-negative plane $\{\}h|\sum h_i= \sum h_0,h_i\ge 0\}$. So, the asymptotic stability is relative to this plane. From the structure of $U_{ij}$, the closer the water level is to the horizontal, the more stringent the conditions for maintaining the stability of the numerical algorithm \ref{eq:3.4}, and the smaller the corresponding time step.
  
Simulation parameters of our model. In this experiment, each grid in DEM is regarded as a cell, with a total of 51,200 cells and with 200 rows and 256 columns. The in equation \ref{eq:4.3} is taken as 0.05 ${m/s}^{1/3}$ \cite{guo2023}, $\Delta x_{ij}$ is taken as 2m which is the distance between the centroid of the two-cell. The calculated $k$ value is 1.25.

The boundary is set to closed and the initial water depth is 0. Hydrologic processes such as infiltration and interception are not considered. In this experiment, $F\left(h\right)$ is not differentiable at the fixed point, so the Jacobian matrix at the fixed point cannot be obtained, and then the corresponding eigenvalues cannot be obtained to estimate the value of $Q$. Based on the calculated $k$, sufficiently large $q$ is chosen in this paper. Whatmore, we take $q=100$, $h=1s$ and the simulation period  $1.5h$ (5400s).

To illustrate the effectiveness of the proposed method, our model is compared with HEC-RAS (version 6.3) software. HEC-RAS developed by the U.S. Army Corps of Engineers is designed to perform one and two-dimensional hydraulic calculations. HEC-RAS simulation parameters are set as follows. The HEC-RAS grid size is 2m$\times$2m, which basically corresponds to the DEM grid one by one, that is, to the cell in our model one by one. Manning’s coefficient and rainfall parameters are the same as those set by our model. Eulerian Shallow Water Equation Solver (SWE-EM (stricter momentum)) is selected for the calculation method, and the default algorithm parameters are applied for the SWE-EM method.

For comparison, the distance $d_p^\prime$ between two water depth snapshots (the depth vector at a certain moment) is defined as: $d_p^\prime = \sum_i p_i |h_i^1-h_i^2|$, where $|h_i^1-h_i^2|$ is the absolute value of water depth difference of the $i-th$ cell between the snapshot 1 and of the snapshot 2, $p_i$ is the probability (frequency) of  $|h_i^1-h_i^2|$ . The distance $d_p^\prime$ is equivalent to $L_1$ norm divided by the total number of cells. $L_2$ norm divided by the total number of cells denoted as $d_2^\prime$ is also used to measure the distance between two water depth snapshots. Let $\hat h$ to be the mean water depth of the entire cellular space. To measure the relative error, the relative distance $R_p$ corresponding to $d_p^\prime$ is defined as: $R_p=d_p^\prime/ \bar h$, and the relative distance $R_2$ corresponding to $d_2^\prime$ is defined as: $R_2=d_2^\prime/ \bar h$.

The comparison of simulation results between our model and the HEC-RAS model is shown in Fig \ref{fig4}. Water depth snapshots at 20 min, 60 min and 90 min are selected for comparison, corresponding to the early, middle and late flood respectively. The average water depth of the total.

\begin{center}
\begin{table*}[!h]%
\caption{The comparative analysis between our model and HEC-RAS model.\label{tab1}}
\begin{tabular*}{\textwidth}{@{\extracolsep\fill}ccccccc@{}}
\toprule
 &  &  &  &  & \multicolumn{2}{c}{\textbf{Total water $(m^3)$}} 
 \\\cmidrule{6-7}
time(min)  & $d_p^\prime$ (mm) & $d_2^\prime$ (mm)  & $R_p$  & $R_2$ & Our model & HEC-RAS  \\
\midrule
20 & 2.2  & 0.02  & 7.8\%  & 0.07\% & 5939.48 & 5871.16   \\
60 & 2.1  & 0.018  & 7.1\%  & 0.06\% & 5938.82 & 5837.39   \\
90 & 2.1  & 0.017  & 7.3\%  & 0.05\% & 5938.48 & 5832.83   \\
\bottomrule
\end{tabular*}
\end{table*}
\end{center}

rainfall is 0.029m, and the maximum water depth when stable is roughly 0.38m to 0.4m. At each selected moment, the distances and the relative distances between the two model depth snapshots are shown in Tab 1. The distance $d_p^\prime$ is around 2mm, and the relative distance $R_p$ is around 7.5\%, which is 7.8\%, 7.1\% and 7.3\%, respectively. The distance $d_2^\prime$ is two orders of magnitude smaller than the distance $d_p^\prime$, which reflects that the distance distribution on the cells field is relatively uniform to some extent. It is worth noting that the water depth distribution in HEC-RAS snapshots is terraced, while the water depth distribution in our model is smoother. The above analysis shows that the output results of the two models are basically consistent, which confirms the effectiveness of our model in solving the actual urban flooding problem.

\begin{figure}[htbp!]
	\centering  
	\begin{subfigure}{0.48\textwidth}  
		\centering
		\includegraphics[height=15pc]{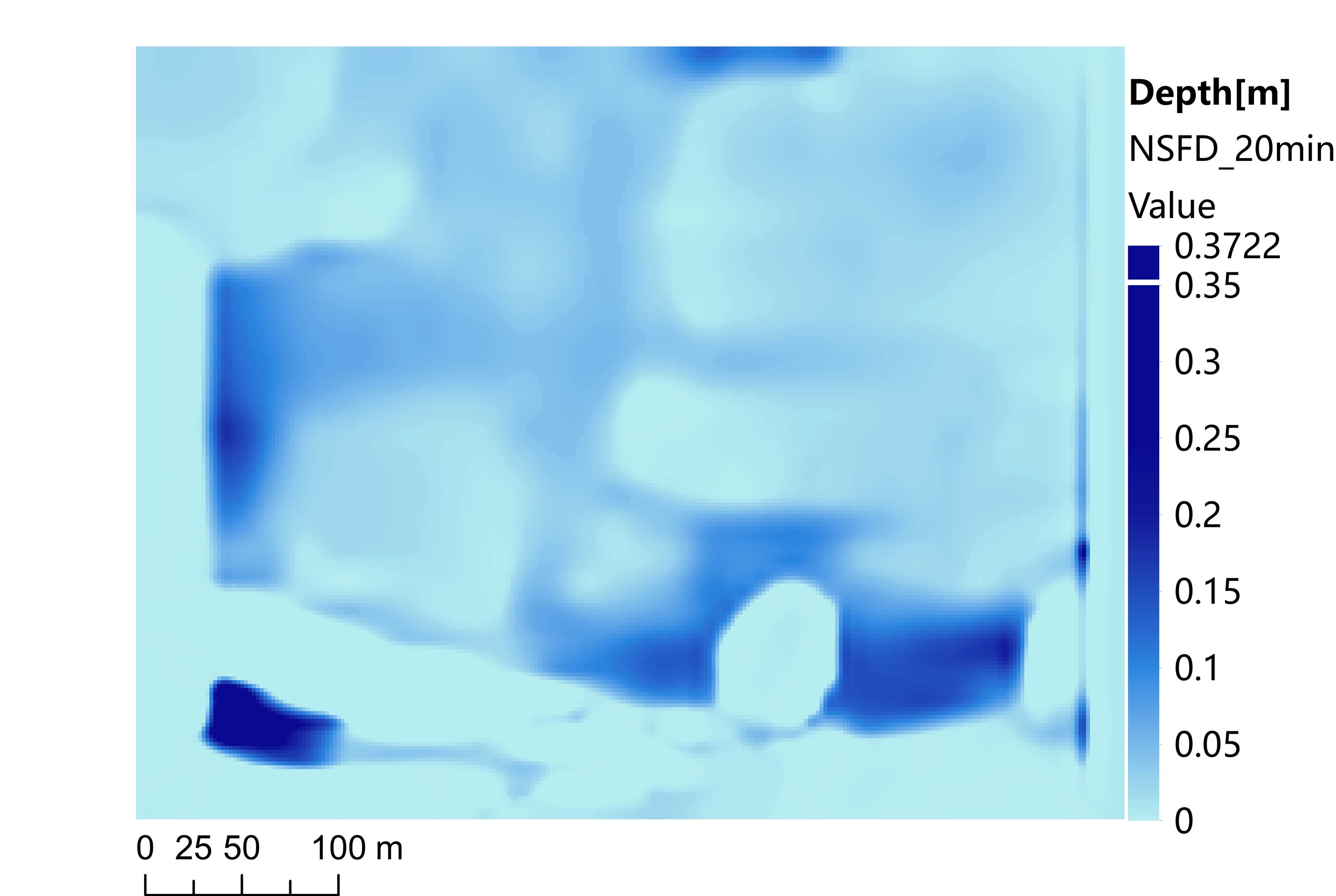}
	\end{subfigure}
	\hfill 
	\begin{subfigure}{0.48\textwidth}
		\centering
		\includegraphics[height=15pc]{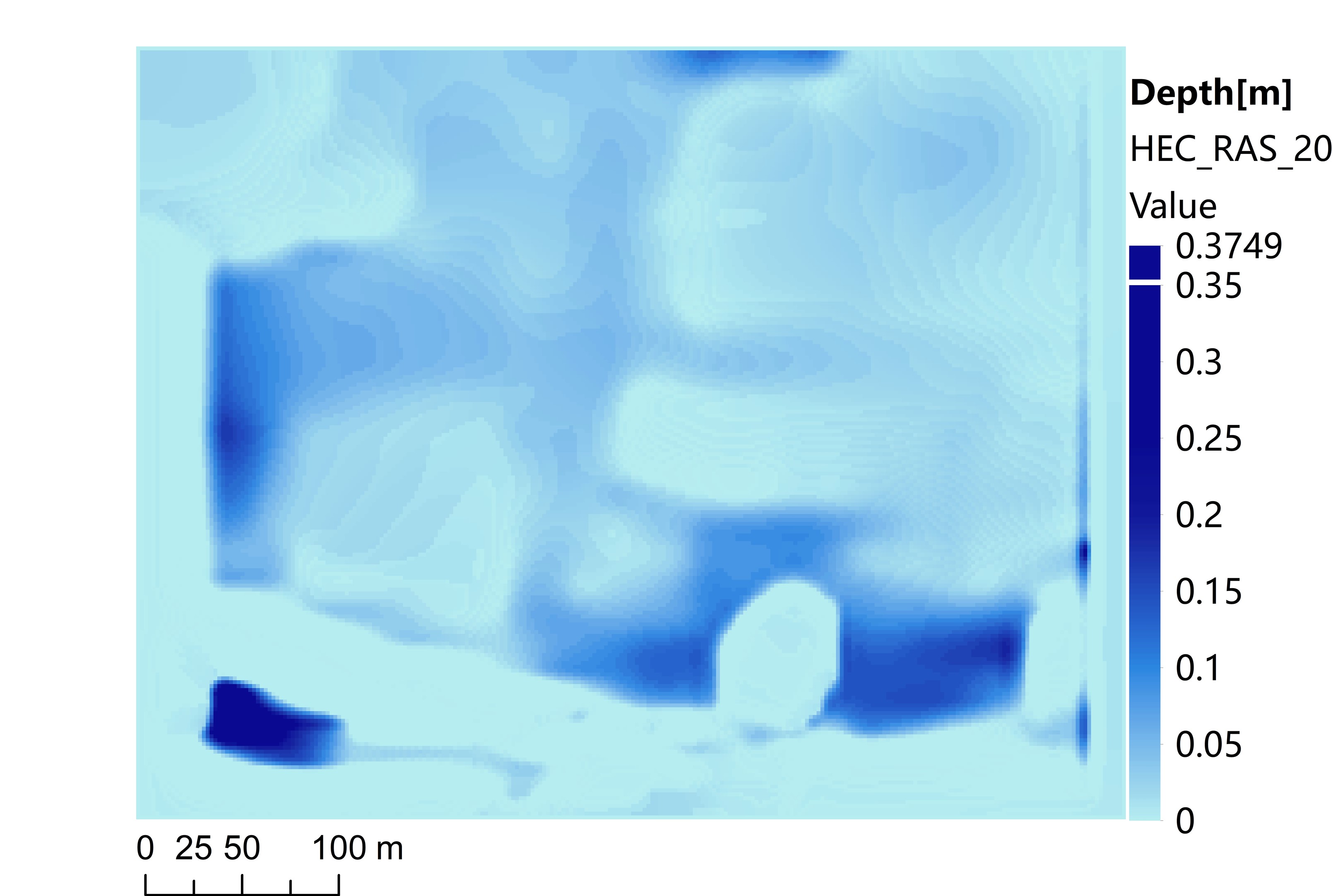}
	\end{subfigure}
	
	\begin{subfigure}{0.48\textwidth}
		\centering
		\includegraphics[height=15pc]{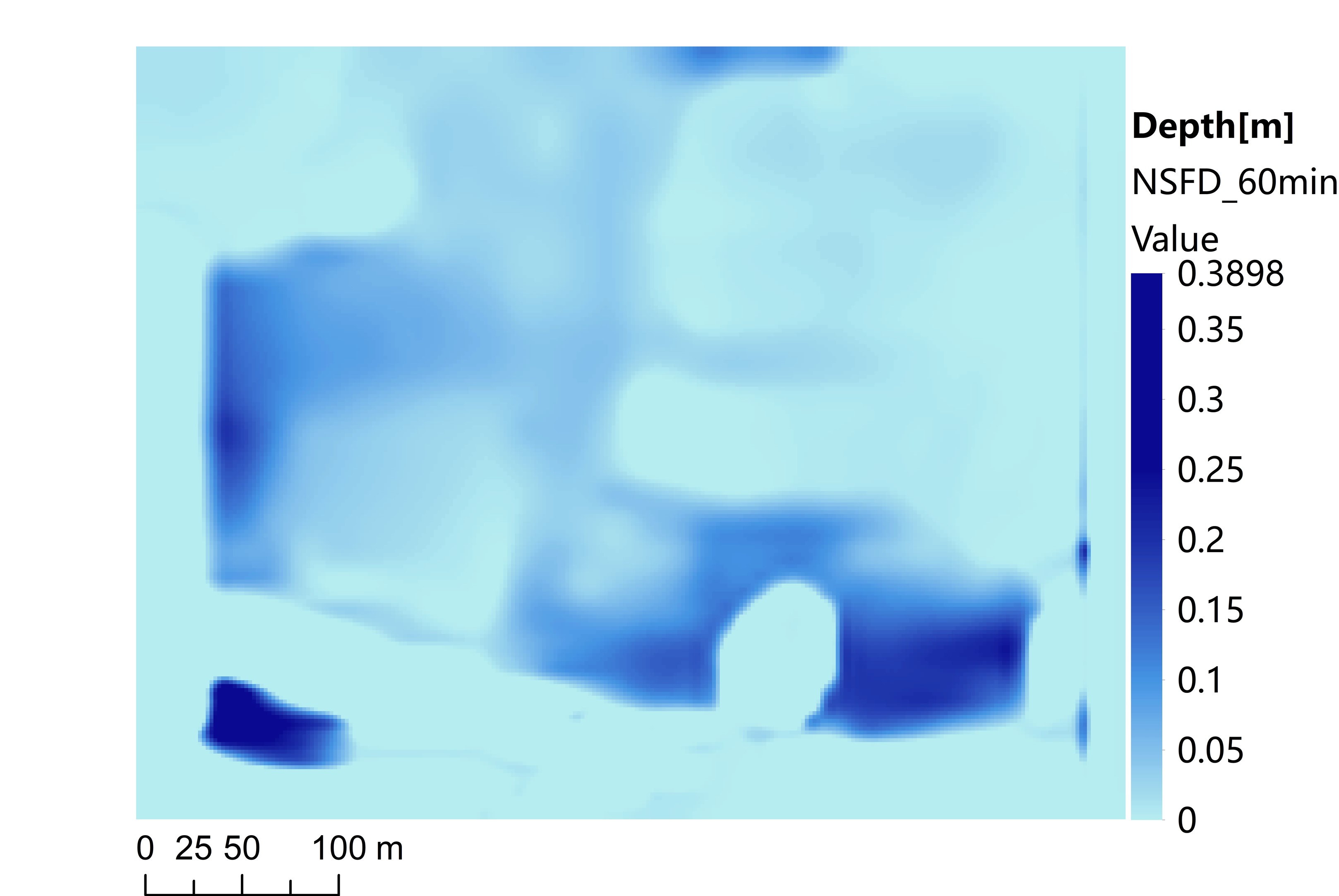}
	\end{subfigure}
	\hfill
	\begin{subfigure}{0.48\textwidth}
		\centering
		\includegraphics[height=15pc]{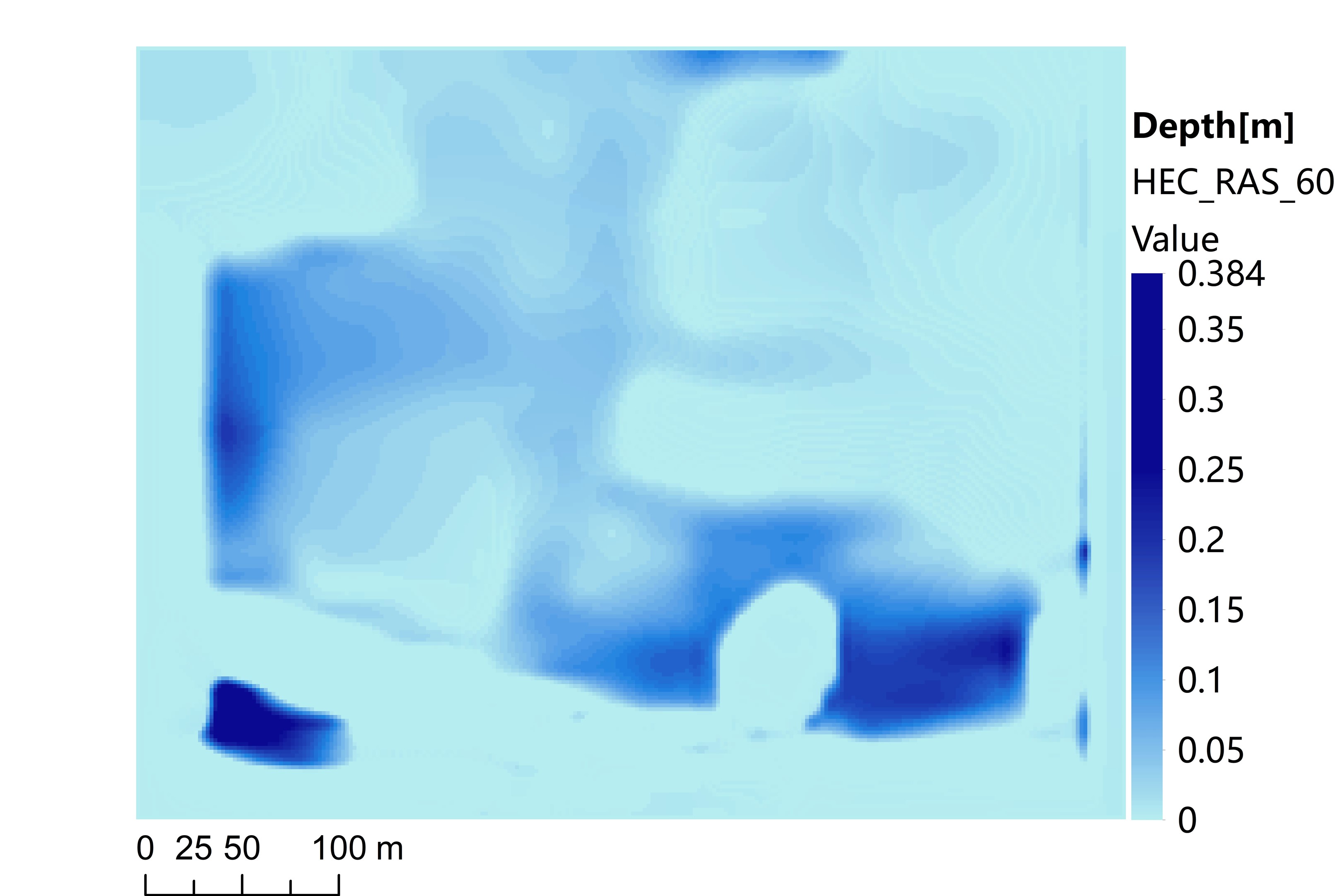}
	\end{subfigure}
	
	\begin{subfigure}{0.48\textwidth}
		\centering
		\includegraphics[height=15pc]{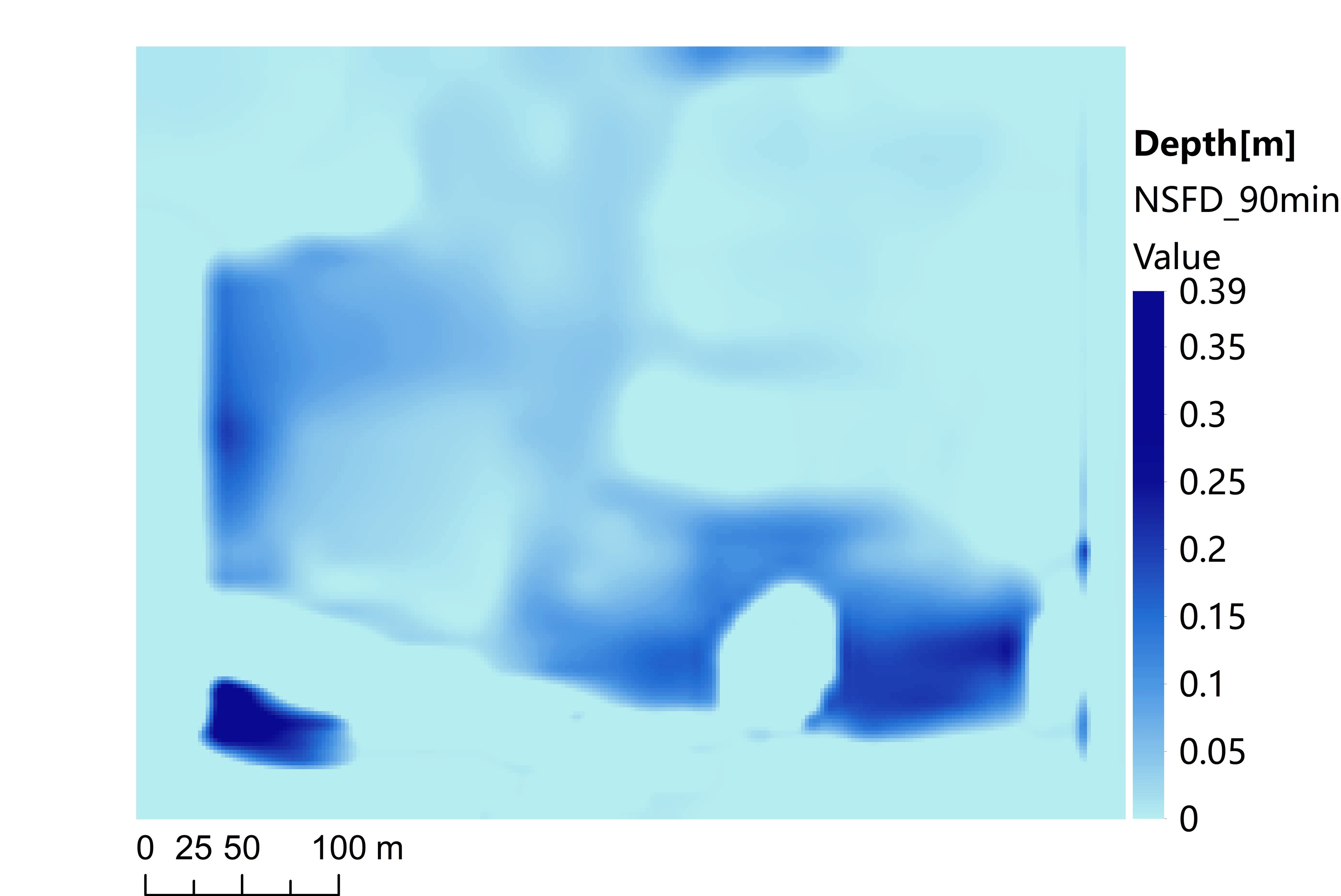}
	\end{subfigure}
	\hfill
	\begin{subfigure}{0.48\textwidth}
		\centering
		\includegraphics[height=15pc]{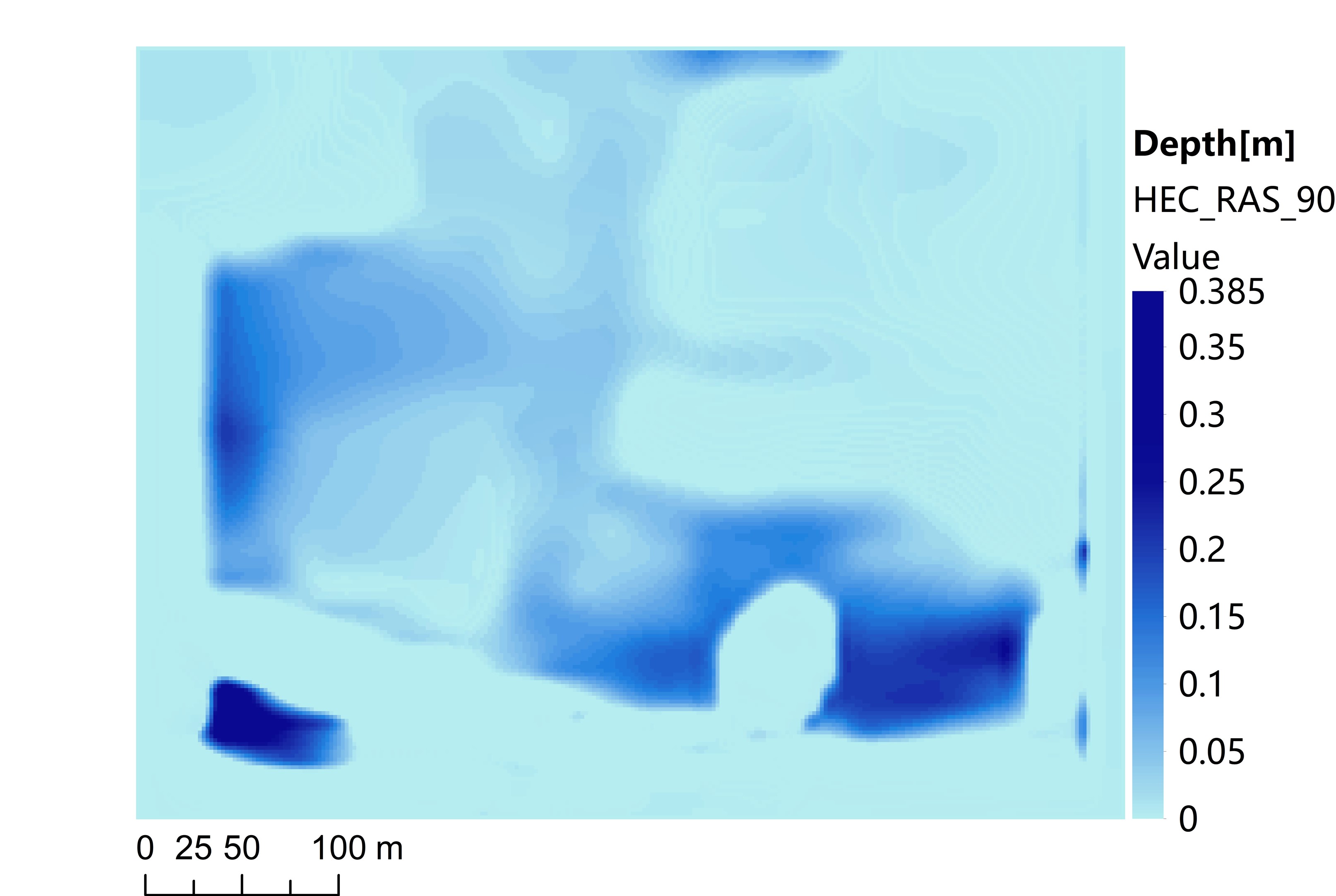}
	\end{subfigure}
	\caption{ The water depth snapshots at 20min, 60min and 90min from our model and HEC-RAS respectively. \label{fig4}}
\end{figure}

The total rainfall volume in the region is $P=5,939.2m^3$, and the rainfall ended at 9 minutes. Considering the boundary closure and water conservation, the total water volume should be the same as the total rainfall volume at any time after 9 minutes. That is, after 9 minutes, the depth vector only moves in the non-negative plane $\{ h | \sum h_i s = P, h_i \ge 0\}$. According to the data in Tab 1, the conservation of water is basically maintained by both the model and the SWE-EM method in HEC-RAS. This also shows that the numerical algorithm \ref{eq:3.4} does have conservation characteristics. Whether water is basically conserved could reflect the error propagation in the numerical calculation process. It should be noted, however, that this "reflect" is not sufficient. In the complex dynamical system of the experiment, the water is still conserved, reflecting that the error is very likely not amplified. Therefore, q and h of this experiment are reasonable.

\section{Conclusions}\label{sec5}

In this paper, the theory of the Urban Flood Dynamic System Model is constructed based on the idea of the Cellular Automata Urban Flood Model. In this model, water interaction rules can be flexibly selected according to the actual urban environment, but should meet the macroscopic physical constraints of water flow movement. In this paper, a constraint condition is given for the cellular water exchange rules matrix. If the condition is satisfied, the water depth evolution naturally satisfies the non-negative, conservation, non-increasing total gravitational potential energy and eventually tends to the invariant distribution. 

To address numerical oscillation, negative water depth, and water non-conservation in the CAUFM, we propose a first-order conservation nonstandard finite difference algorithm to solve the UFDSM by enhancing an existing nonstandard finite difference algorithm. The proposed method ensures flow conservation by separately considering the inflow and outflow of the central cell, treating the outflow from upstream cells as the inflow for central cell, the value both of which are not equal in the non-standard finite difference algorithm referenced. This approach allows for comprehensive flow calculations based solely on cell outflows. We prove that the proposed numerical method has first-order precision, and give the numerical stability condition through the model equation, and illustrate that the numerical stability condition is naturally satisfied as long as the elementary stable condition is satisfied. Finally, by conducting an experiment with an analytical solution, we compare the proposed UFDSM and numerical method against the analytical solution, demonstrating that the numerical method exhibits the properties we have proven. Furthermore, through experimentation in an urban region, we compare our model with HEC-RAS software and demonstrate its ability to accurately simulate the flood process.

\section{Discussion}\label{sec6}

The UFDSM proposed in this paper is suitable for complex urban areas with small topographic relief due to the significant weakening of the momentum term caused by numerous obstacles. However, its applicability may be limited in other areas where the momentum term cannot be effectively weakened. Although the nonstandard finite difference algorithm proposed naturally preserves many properties of the UFDSM, it requires a small time step. Nevertheless, this limitation can be overcome by employing an implicit updating method and dynamically adjusting the step size strategy to increase the time step size, which represents one of our future research directions.

The fundamental advantage of our proposed UFDSM is that it absorbs the data of complex urban systems as much as possible by simplifying the algorithm, while the traditional 2D hydraulic model adapts to the algorithm by simplifying the urban system. Firstly, we can use various cell types to represent different urban land uses such as buildings, roads, grasslands, water bodies, open spaces, mountains, pipeline nodes and boundaries that can be expanded freely. Secondly, different and diverse cellular water flux calculation rules and updating rules can be adopted. Different cellular types may adopt different interaction rules when water interactions occur. For example, the interaction between building cells and other cells may select weir flow formula or pore flow formula, and the interaction between pipeline nodes and cells on the ground can be set as weir flow formula and overflow rule. In addition, the flux rules between any two cells can be customized to achieve a more accurate portrayal of the complex flow of the city. Thirdly, the neighborhood of the cell can be adjusted flexibly, which can realize the characterization of barriers such as dykes and community walls. Compared with the traditional 2D hydrodynamical model, we stand on the data side in the UFDSM.

Compared with the partial differential equation model based on shallow water equation, UFDSM is easier to couple other hydrological and hydraulic processes, such as distributed rainfall model, river network model, drainage network model, infiltration model, hydraulic engineering scheduling model, etc., thereby enabling the development of a comprehensive urban flood and flood process rapid simulation framework.

The UFDSM offers a more accessible approach to data assimilation compared to partial differential equation models, thereby enabling more precise predictions. Given the well-established and mature research on data assimilation in dynamical systems, it can be readily applied to this model. 
Consequently, the UFDSM proposed in this paper exhibits comparative advantages and holds promising prospects for practical applications in simulating urban flood processes.


\bmsection*{Acknowledgments}

The authors declare that they have no known competing financial interests or personal relationships that could have appeared to influence the work reported in this paper.
Funding: This work was supported by National Natural Science Foundation of China [Grant No. 72104123] and Tsinghua University Initiative Scientific Research Program [grant number 20223080006].

\bmsection*{Open Research}

In Part 4.1, the DEM data (Fig \ref{fig1}(a)) and initial water depth data (Fig \ref{fig1}(b)) used in the model equation were well-designed by us. The calculation process for the model equation is fully implemented by python version 3.10.3, available at http://www.python.org, the data structure and some basic algorithms mainly rely on Numpy version 1.22.3 \cite{harris2020}, available under the Numpy license at https://numpy.org/, and the resulting graph (Fig \ref{fig1}(c) and Fig \ref{fig1}(d)) were drawn using matplotlib version 3.5.1 \cite{hunter2007}, available under the Matplotlib license at https://matplotlib.org/. In Part 4.2,the DEM data is generated by the author's modulation. The algorithm process was completely implemented by ISO 14 C++ version within Windows10,  available at https://en.cppreference.com/w/cpp/14. HEC-RAS version 6.3 is available under  CEIWR-HEC's policy at https://www.hec.usace.army.mil/software/hec-ras/. The  Fig\ref{fig2} and Fig\ref{fig3} were drawn using free and open source Geographic Information System QGIS 3.34, available under at https://www.qgis.org/.The code, data, and configuration files for this study are licensed under MIT and published on GitHub https://github.com/tianyongsen/Urban-Flood-Dynamical-System-Model.git.

\bmsection*{Conflict of interest}

The authors declare no potential conflict of interests.

\bmsection*{Data Availability Statement}

Data sharing not applicable to this article as no datasets were generated or analysed during the current study.

\bibliography{wileyNJD-AMA}

\end{document}